\documentclass[10pt,conference,letterpaper]{IEEEtran}
\pdfoutput=1
\usepackage{pgfplots}
\pgfplotsset{compat=newest}
\usepackage{tikz}
\usetikzlibrary{arrows,matrix,positioning,patterns}
\usepackage[utf8]{inputenc}
\usepackage[innermargin=0.75in,outermargin=0.75in,top=0.5in,bottom=0.6in]{geometry}
\usepackage[english]{babel}
\usepackage[T1]{fontenc}
\usepackage{epsfig}
\usepackage{amsmath, amssymb, amsbsy}
\usepackage{mathdots}
\usepackage{xspace}
\usepackage[noend]{algpseudocode}
\usepackage[linesnumbered,ruled,vlined,titlenumbered]{algorithm2e}
\usepackage{algorithmicx}
\usepackage{color}
\usepackage{cite}
\usepackage{booktabs}
\usepackage{verbatim}
\usepackage{url}
\usepackage{lipsum}
\usepackage{enumitem}
\usepackage{colortbl}
\usepackage{amsthm}
\usepackage{dblfloatfix}
\usepackage{verbatim}
\makeatletter
\newcommand\footnoteref[1]{\protected@xdef\@thefnmark{\ref{#1}}\@footnotemark}
\makeatother

\usepackage{multicol}

\usepackage[final,tracking=true,kerning=true,spacing=true,factor=1100,stretch=10,shrink=20]{microtype}

\newenvironment{mymatrix}{\begin{bmatrix}} {\end{bmatrix} }

\def\ve#1{{\mathchoice{\mbox{\boldmath$\displaystyle #1$}}
              {\mbox{\boldmath$\textstyle #1$}}
              {\mbox{\boldmath$\scriptstyle #1$}}
              {\mbox{\boldmath$\scriptscriptstyle #1$}}}}

\newcommand{\mycode}[1]{\ensuremath{\mathcal{#1}}}

\renewcommand{\H}{\ve{H}}

\renewcommand{\P}{\ve{P}}

\newcommand{\F}{\mathbb{F}}

\usepackage{graphics}
\usepackage{subcaption}
\usepackage{epsfig} 
\usepackage{times} 
\usepackage{amsmath}
\usepackage{mathtools}
\usepackage{amssymb}  
\usepackage{cite}
\usepackage{soul}
\usepackage{pgfplots}
\usepackage{pgfplotstable}
\usepackage[mathscr]{eucal}
\usepackage{listings}

\newtheorem{lem}{Lemma}
\newtheorem{prop}{Proposition}
\newtheorem{examplex}{Example}
\newtheorem{cor}{Corollary}

\newtheorem{const}{Construction}
\newtheorem{thmmystyle}{Theorem}
\newtheorem{rem}{Remark}
\usepackage{tikz}
\usetikzlibrary{shapes}
\usetikzlibrary{shapes.multipart,chains}
\makeatletter
\newcommand{\removelatexerror}{\let\@latex@error\@gobble}

\newcommand*\fullcirc[1][1ex]{\tikz\fill (0,0) circle (#1);}
\newcommand*\emptycirc[1][1ex]{\tikz\draw (0,0) circle (#1);} 

\newcommand\bigzero{\makebox(0,0){\text{\huge0}}}
\newcommand{\Pdmc}{Defective }
\newcommand{\pdmc}{defective }
\newcommand{\psmc}{PDMC }
\makeatother
\makeatletter
\newcommand*{\rom}[1]{\expandafter\@slowromancap\romannumeral #1@}
\makeatother

\IEEEoverridecommandlockouts
\pgfkeys{
	/tr/rowfilter/.style 2 args={
		/pgfplots/x filter/.append code={
			\edef\arga{\thisrow{#1}}
			\edef\argb{#2}
			\ifx\arga\argb
			\else
			
			\fi
		}
	}
}
\usepackage{filecontents}

\begin{document}
\title{Codes for Preventing Zeros at \\Partially \Pdmc Memory Positions\vspace{-0.4 cm}}
\author{\IEEEauthorblockN{Haider Al Kim$^{1,2}$\thanks{
This work has received funding from the German Academic Exchange Service (Deutscher Akademischer Austauschdienst, DAAD) under the support program ID 57381412, and the European Union's Horizon 2020 research and innovation program through the Marie Sklodowska-Curie under Grant No.~713683 and No.~801434.
}, Kai Jie Chan$^{3}$, }
\IEEEauthorblockA{
  $^1$Institute for Communications Engineering, Technical University of Munich (TUM), Germany\\ $^2$Electronic and Communications Engineering, University of Kufa (UoK), Iraq\\
  $^3$Singapore Institute of Technology (SiT) and TUM-Asia, Singapore\\
  Email: haider.alkim@tum.de, CHAN KAI JIE <1802319@sit.singaporetech.edu.sg>}
  \vspace{-1cm}}
\maketitle
\begin{abstract}
This work deals with error correction for non-volatile memories that are partially defective at some levels. Such memory cells can only store incomplete information since some of their levels cannot be utilized entirely due to, e.g., wearout. On top of that, this paper corrects random errors $t\geq 1$ that could happen among $u$ partially defective cells while preserving their constraints. 
First, we show that the probability of violating the partially defective cells' restriction due to random errors is not trivial. Next, we update the models in \cite{haideralkim2019psmc} such that the coefficients of the output encoded vector plus the error vector at the partially \pdmc positions \emph{are non-zero}.
Lastly, we state a simple proposition (Proposition~\ref{Proposition_3}) for masking the partial defects using a code with a minimum distance $d$ such that $d\geq 2(u+t)+1$. "Masking" means selecting a word whose entries correspond to writable levels in the (partially) defective positions.
A comparison shows that masking $u$ cells by this proposition for a particular BCH code is as effective as using the complicated coding scheme proven in \cite[Theorem~1]{haideralkim2019psmc}.
\end{abstract}
\begin{IEEEkeywords}
flash memories, phase change memories, non-volatile memories, defective memory, (partially) \pdmc cells, stuck-at errors, error correction codes, BCH codes \vspace{-0.1cm}
\end{IEEEkeywords}
\section{Introduction}
The growing demand for energy-efficient memory solutions that exhibit short processing times has fueled the adoption of non-volatile memory. 
Research has shown that phase change memories (PCMs), a type of non-volatile memory, are highly efficient, thus making them a viable replacement for current storage technology such as DRAM in the foreseeable future. A distinctive feature of a PCM cell is its ability to alter between two main states, namely, an amorphous state and a crystalline state, which also directly correspond to the OFF and ON states or logic '0' and '1', respectively. The crystalline state is further defined by its multi-programmed levels. Due to degradation caused by the heating and cooling processes of the cell, PCMs may be unable to change their states. In such an event, the cell is termed as a \emph{\pdmc memory cell} since it can only store a single phase. In multi-level PCMs, e.g., dual PCMs, due to the thermal processes of the cells, failure may happen at a position in between both primary states or in the partially crystalline levels. Thus, the cell is called a \textit{partially \pdmc memory cell} \cite{Gleixner2009,Kim2005,Lee2009,Pirovano2004}.
In flash memory, different amounts of charge are used to represent the various levels within the cells. If one wants to overwrite the stored information, one could either raise the \pdmc\unskip-at level or reset the whole cell to its original amorphous state. Since the latter method cuts the lifespan of such memory devices, this paper considers only raising the level using a mechanism called \textit{masking}. This technique ensures that memory is correctly utilized in the faulty cells by checking that a codeword fits the \pdmc level within the cells. 
\subsection{Related Work}
Code constructions capable of masking \pdmc memory cells and correcting errors during the storing and reading procedure are put forth in \cite{heegard1983partitioned}. However, \cite{heegard1983partitioned} considers classical defects and errors correcting code constructions, so the author suggested a relatively large redundancy, i.e., the required check symbols are at least the number of defects.  
Later, reductions in the redundancy needed for masking \textit{partially} \pdmc cells are achieved in \cite{wachterzeh2016codes}. Nevertheless, this paper does not regard error correction on top of masking. 
A more recent scenario in \cite{haideralkim2019psmc} considers code constructions for simultaneous masking of partially \pdmc cells and error correction by synthesizing techniques from \cite{heegard1983partitioned} and \cite{wachterzeh2016codes}. 
The code constructions in \cite{haideralkim2019psmc} consider a scenario of $q$ levels memory cells in which all partially defective cells are stuck at the level $1$, so zeros are forbidden in these positions. They are formed under the supposition that when random errors occur in the partially \pdmc cells, they concede to the partially \pdmc constraints.  
However, that is not always guaranteed to happen since it is merely an \textit{idealized} assumption. Let $c_i$ be the resulting coordinate, after the encoding process in \cite[Theorem~1]{haideralkim2019psmc}, at the $i$-th position where the partially defective cell is, and let $e_j$ be an error value that happens at location $j=i$. The authors in \cite{haideralkim2019psmc} then assume that $c_i+e_j \neq 0$. 
Ignoring this assumption and given that the calculations are done in the finite field $\mathbb{F}_q$, we end up with $c_i+e_j=0$ for $c_i=q-1$ and (coincidentally) $e_j=1$ 
if $e_j$ occurs 
before writing to the memory of partial defects (see Figure \ref{fig:blockdiagramwithlegend}). Therefore, although the encoding algorithm successfully provides a vector that matches the partial defects, the storing process might fail to present a vector that can be properly placed on that memory due to random errors, or the reading process might be unsuccessful due to \emph{mag-1} error (magnitude error as defined in Section~\ref{mag-1} ) \cite{ASolomon}.
\begin{figure}[h]
	\begin{subfigure}[b]{\linewidth}
	\scalebox{0.7}{
	\begin{tabular}{l*{1}{c}}
		Legend    &        \\
		\hline
		$\ve{m}$:  message vector & \\
		$\ve{G}_1$:   generator matrix &\\
		$\ve{w} = \ve{m} \cdot \ve{G}_1$: augmented message vector &  \\
		$\ve{\phi}$: partially stuck positions, where $i \in \ve{\phi}$ &  \\
		$\Psi$: error positions, where $j \in \Psi$ & \\
		$\ve{d} = z_0\cdot \ve{G}_0$: masking vector &  \\
		$\ve{c} = \ve{w}+ \ve{d}$: output codeword &    \\
		$\ve{e}$: error vector &   \\
		$\ve{y} = \ve{c}+\ve{e} $: corrupted codeword & 
	\end{tabular}
}
	\end{subfigure}
	\begin{minipage}[b]{0.4\textwidth}
		\begin{subfigure}[b]{\linewidth}
			\begin{center} 
				\scalebox{0.65}{ \begin{tikzpicture}[auto, thick, >=triangle 45]
    \node[draw, thick, rectangle, minimum height = 2.5em, minimum width = 2.5em] at (1.4,0) (matrix) {$\boldsymbol{G}_1$};
    \node[draw, circle, minimum size=0.2cm] (adder1) at (3.3,0){};
    \draw (adder1.east) -- (adder1.west) (adder1.north) -- (adder1.south);
    \node[draw, circle, minimum size=0.2cm] (adder2) at (5.8,0){};
    \draw (adder2.east) -- (adder2.west) (adder2.north) -- (adder2.south);
    
    \coordinate (message) at (0,0);
    \coordinate (masking) at (3.3,1.4);
    \coordinate (error) at (5.8,1.4);
    \coordinate (corrupted) at (8.2,0);
    
    \draw[-latex](message) -- node {$\boldsymbol{m}$}(matrix);
    \draw[-latex](matrix) -- node {$\boldsymbol{w},w_{i}$}(adder1);
    \draw[-latex](adder1) -- node {$c_{i}\in\mathbb{F}_q\backslash\{0\}$}(adder2);
    \draw[-latex](masking) -- node {$d_{i}$}(adder1);
    \draw[-latex](error) -- node {$e_{j}\in\{0,1\}$}(adder2);
    \draw[-latex](adder2) -- node {$y_{i,j}\in\mathbb{F}_q\backslash\{0\}$}(corrupted);
    
    \draw[color=gray,thick](0.7,-0.7) rectangle (4,1.1);
	\node at (0.5,1.3) [above=5mm, right=0mm] {\textsc{encoder}};
	
	\draw[dashed] (2.4,-1) -- (2.4,1.7);
    \end{tikzpicture}}
				\caption{Assumption $c_i+e_j \neq 0$ in \cite{haideralkim2019psmc}.}
				\label{fig:blockdiagramA}
			\end{center}
		\end{subfigure}
		\\[\baselineskip]
		\begin{subfigure}[b]{\linewidth}
			\begin{center}
				\scalebox{0.65}{   \begin{tikzpicture}[auto, thick, >=triangle 45]
    \node[draw, thick, rectangle, minimum height = 2.5em, minimum width = 2.5em] at (1.4,0) (matrix) {$\boldsymbol{G}_1$};
    \node[draw, circle, minimum size=0.2cm] (adder1) at (3.3,0){};
    \draw (adder1.east) -- (adder1.west) (adder1.north) -- (adder1.south);
    \node[draw, circle, minimum size=0.2cm] (adder2) at (6.8,0){};
    \draw (adder2.east) -- (adder2.west) (adder2.north) -- (adder2.south);
    
    \coordinate (message) at (0,0);
    \coordinate (masking) at (3.3,1.4);
    \coordinate (error) at (6.8,1.4);
    \coordinate (corrupted) at (9.2,0);
    
    \draw[-latex](message) -- node {$\boldsymbol{m}$}(matrix);
    \draw[-latex](matrix) -- node {$\boldsymbol{w},w_{i}$}(adder1);
    \draw[-latex](adder1) -- node {$c_{i}\in\mathbb{F}_q\backslash\{0,q-1\}$}(adder2);
    \draw[-latex](masking) -- node {$d_{i}$}(adder1);
    \draw[-latex](error) -- node {$e_{j}\in\{0,1\}$}(adder2);
    \draw[-latex](adder2) -- node {$y_{i,j}\in\mathbb{F}_q\backslash\{0\}$}(corrupted);
    
    \draw[color=gray,thick](0.7,-0.7) rectangle (4,1.1);
	\node at (0.5,1.3) [above=5mm, right=0mm] {\textsc{encoder}};
	
	\draw[dashed] (2.4,-1) -- (2.4,1.7);
    \end{tikzpicture}}
				\caption{Guarantee $c_i+e_j \neq 0$ in this work.}
				\label{fig:blockdiagramB}
			\end{center}
		\end{subfigure}
	\end{minipage}
	\caption{Encoding process showing the masking and the disturbance stages.}
	\label{fig:blockdiagramwithlegend}
	\vspace{- 0.5cm}
\end{figure}
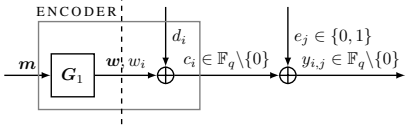
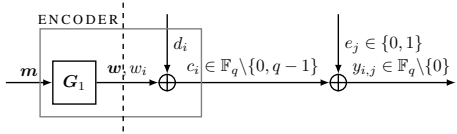
\subsection{Our Contribution}
In response to the aforementioned limitation, we want to ensure that the partially \pdmc constraint is always met by fulfilling $\{\forall i:1\leq c_i\leq q-2\}$, given the assumption that the error vector $\boldsymbol{e} \in \{0,1\}^n$. In this paper, we first show that the probability of the reverse assumption in \cite{haideralkim2019psmc} (the output vector coordinates at the partially stuck positions become zeros due to random errors) is non-trivial. Then, modified methods from \cite{haideralkim2019psmc} to obtain code constructions for jointly masking partially \pdmc cells and correcting errors have been conducted. This modification ensures that if the errors \emph{jointly happen at the partially \pdmc positions}, the partially \pdmc constraint is always satisfied.
The price of relaxing the restriction from \cite{haideralkim2019psmc} is the reduced $u$ of the updated versions of code constructions that this paper
can deal with, is precisely half compared to the constructions proven in \cite{haideralkim2019psmc}.
Although our method comes at the cost of masking capability as less number of $u$ cells can be treated than \cite{haideralkim2019psmc},
our work preserves the required redundancies as achieved by \cite{wachterzeh2016codes} and \cite{haideralkim2019psmc} for masking-only and masking-and-error-correcting, respectively, and it
is more realistic regarding physical memories with partial defects that also suffer from random errors.
\section{Preliminaries}
\subsection{Notations}\label{ssec:notation}
For a prime power $q$, $\mathbb{F}_{q}$ denotes the finite field of order $q$. 
To coincide with the notations in \cite{haideralkim2019psmc}, we let $k$ be the number of information symbols, $l$ be the required symbol(s) for masking, $r$ be the required redundancy for error correction, $t$ be the number of errors, and $u$ be the number of (partially) \pdmc cells. Let $d$ be the code minimum distance, and $n$ be the code length and also the memory size. 
$s$ denotes the (partially) \pdmc level at any position. 
In general, for positive integers $h,f$ we denote by $[h]$ the set of integers $\{0,1,\dots,h-1\}$ and $[h,f]=\{h,h+1,\dots,f-1\}$. Vectors and matrices are denoted by lowercase and uppercase boldface letters, e.g. $\boldsymbol{a}$ and $\boldsymbol{A}$, and are indexed starting from $0$. $\text{RRE}(\boldsymbol{A}^{(i)})$ denotes the reduced row Echelon form of a matrix $\boldsymbol{A}^{(i)} $
that has its columns indexed by $i$.
Note that all calculations are done in the finite field $\mathbb{F}_q$.
\subsection{Definitions}\label{ssec:definitions}
\subsubsection{Partially Defective Cells}
If a cell is unable to switch its value and always stores the value $s\in[q]$, the cell is termed as being \emph{\pdmc at level $s$}. On the other hand, if a cell is only able to store values greater than or equal to $s\in[q]$, it is termed as being \emph{partially \pdmc at level $s$}. A healthy cell which can store any of the $q$ levels is equivalent to a cell that is partially \pdmc at 0.
\subsubsection{Error Model}
Among $n$ total cells, let there be $u$ partially \pdmc\unskip-at-1 cells at positions $\ve{\phi}=\{\phi_0,\phi_1,\dots,\phi_{u-1}\}\subseteq[n]$.  
Let the set $\Psi=\{\Psi_0,\Psi_1,\dots,\Psi_{t-1}\}\subseteq[n]$ denote the positions of $t$ errors introduced by the channel.
In this work, we address an \textit{overlapping model} in which random errors can happen in any of $[n]$ positions. 

We use $a_{ov} \in \mathbb{F}_q$ to indicate a value in an overlapped position in which $\{\ve{\phi} \cap \Psi \neq \emptyset \,\mid \, \mbox{for } i \in \ve{\phi} \mbox{ and } j \in \Psi \, , \, i =j \} $. We denote by $\ve{e} \in \mathbb{F}^n_q$ as 
all error vectors of Hamming weight $wt(\ve{e})\leq t$
and for any of their coordinates by $e_j$ for $j \in \Psi $.
\subsubsection{Error Type}\label{mag-1}
A dominant error type in non-volatile memory is \emph{mag-1} error (magnitude error) in which for any symbol $x \in \F_q$ that was written to a memory cell, the cell suffers \emph{mag-1} if the read process returns $x' \in \F_q$ such that $\mid x' -x \mid =1$. Hence by this definition, $e_j \in \{-1,1\}$ (equivalently $e_j \in \{q-1,1\}$ as $-1 \mod q \equiv q-1 \mod q$) for $j \in \Psi$.
For this paper the considered error vector $\boldsymbol{e}$ is of coordinates $\in\{0,1\}^n$ (but extendable to other values, cf. Remark~\ref{Rem_extended}) such that $x' = x + e_j \in \{x, x+1\}$ for $j \in \Psi$. Then $(x+e_{j}) \mod{q} =0$ if and only if $x = q-1$ and $e_{j} =1$.
\subsubsection{$(n,M)_q$ $(u,t)$-PDMC code $\mathcal{C}$}
An $(n,M)_q$ $(u,t)$-PDMC code $\mathcal{C}$ is a \emph{partially defective at masking code} which is a coding system with an encoder $\mathcal{E}$ and decoder $\mathcal{D}$. The input of the encoder $\mathcal{E}$ includes
\begin{itemize}
    \item the set of locations $\ve{\phi}$ for $u$ partially \pdmc cells,
    \item the partially \pdmc levels 
    $s_i=1$, for all $i\in\ve{\phi}$,
    \item a message $\ve{m}\in\mathcal{M}$, with $\mathcal{M}$ being a message space of size $\vert\mathcal{M}\vert=M$.
\end{itemize}
It produces a codeword $\boldsymbol{c}\in\mathbb{F}^n_q$ which obeys $\{\forall i\in\ve{\phi}\, , \, \forall j\in \Psi:(c_{i}+e_{j})_{ov}\geq s_i\}$. The decoder $\mathcal{D}$ maps input $\boldsymbol{c}+\boldsymbol{e}$ to the correct message $\boldsymbol{m}$ 
, i.e., $\boldsymbol{e}$ adds distance $\leq t$ to any codeword. We say $u$-\psmc when $t=0$ (masking only) \cite{wachterzeh2016codes}.

\section{Probability of Overlapping Errors Causing Zero Coordinates}
In the following sequel, we present the probability of $t$ errors in the \textit{overlapping model} that is in our interest for the entire work of this paper. 
Then in the subsequent section, we consider how to accommodate errors if they coincidentally happen in the positions where partially defective-at-1 cells are such that \textit{non-zero} occurrences in these positions are guaranteed.
\begin{prop}
	\label{prop_00}(Probability of Overlapping Positions)
	Let $n$ be a positive integer, $\ve{\phi} \subseteq [n]$ have size $u$, and  $\Psi \subseteq [n]$ have size $t$
	such that $u,t \leq n$. 
	The probability of $\ve{\phi} \cap \Psi \neq \emptyset$ is
	\begin{align}\label{eq_overlapping0}
	\P(u,t|n) =   & 1-\prod_{j=0}^{t-1}\frac{n-u -j}{n-j} .
\end{align}
\end{prop}
\begin{proof}
Let $\P(u,t|n)$ be the probability of $\ve{\phi} \cap \Psi \neq \emptyset$.
The proof follows the pigeon-hole principle. 
First, assume there are $n$ empty boxes and $u$ items. We want to put one item from $u$ in each box from $n$, i.e., one-to-one correspondence. 
Since $1 \leq u \leq n$, there are at most $n-u$ empty boxes left after filling 
with $u$ items. 
Now, for new items $1 \leq t \leq n$, there are only $n-u$ empty (non-overlapping) boxes.
Hence, the probability (denoted by $\P(t|n-u)$) that $t$ occupies only $n-u$ empty boxes is
\begin{equation}\label{eq_non_overlapping}
     \P(t|n-u) = \frac{n-u}{n} \times \frac{n-u-1}{n-1} \times \dots  \times \frac{n-u-(t-1)}{n-(t-1)}.  
    \end{equation}
Then by the union bound on the probabilities, $1-\P(t|n-u)$ is $\P(u,t|n)$
which is the expression from \eqref{eq_overlapping0}.  
\end{proof}
\begin{prop}
	\label{prop_01}(Probability of Zero Occurrence in Overlapped Positions)
	Let $n, q$ be positive integers. 
	Let $\ve{\phi} \subseteq [n]$ have size $u$ and  $\Psi \subseteq [n]$ have size $t$
	such that $u,t \leq n$.
	Then for $\ve{c} \in \F^n_q$ and $\ve{e} \in \F^n_q$, the probability of 
	\begin{align}\label{zero_occurrence}
		\Big\{(c_{i}+e_{j})_{ov}\mod q = 0 \,\big|\,\ve{\phi} \cap \Psi \neq \emptyset \,, \, \mbox{for}\,i \in \ve{\phi} \,\mbox{and}\,j \in \Psi\Big\}
	\end{align}
	is
\begin{align}\label{eq_overlapping0_2}
	\P(u,t|n,q) =   & \frac{1}{q}\Bigg( 1-\prod_{j=0}^{t-1}\frac{n-u -j}{n-j}\Bigg).
\end{align}
\end{prop}
\begin{proof}
	Let $ \P(c_i,e_j| q)$ be the probability of a pair $(c_i,e_j)$ satisfying $(c_{i}+e_{j})_{ov} \mod q = 0$.
 
	In each overlapped position in which $\ve{\phi} \cap \Psi \neq \emptyset$, the value is $(c_{i}+e_{j})_{ov} \mod q$.  
	There are $q^2$ total combinations of a pair of the choices of $(c_i,e_j) \in [q] \times [q]$, and there are exactly $q$ choices such that $(c_i+e_j)_{ov} \mod q = 0$. Therefore, the probability $ \P(c_i,e_j| q) = \dfrac{q}{q^2}$.
	By Proposition~\ref{prop_00}, the overlapping happens with probability given in \eqref{eq_overlapping0} and both events $\P(c_i,e_j| q)$ and $\P(u,t|n)$ are independent, so the probability (denoted by $\P(u,t|n,q)$) that condition \eqref{zero_occurrence} occurs is
	\[\P(u,t|n,q) = \P(c_i,e_j| q) \cdot \P(u,t|n)\] which is the expression in \eqref{eq_overlapping0_2}.
\end{proof}
\begin{rem}
If $u+t > n$, $\P(u,t|n) = 1$ in Proposition~\ref{prop_00} which means at least one position in $n$ must be overlapped. Consequently, $\P(u,t|n,q) = \dfrac{1}{q}$ from Proposition~\ref{prop_01}. Hence, for small values of $q\geq 3$, the probability from \eqref{eq_overlapping0_2}
is non-trivial,
i.e., for $q =3$, it is a one-third and for $q=4$ it is a quarter. 
\end{rem}
\section{Codes for Preventing Zeros at Partially \Pdmc Positions Due to Random errors , $u\leq\lfloor\frac{q-1}{2}\rfloor, s_i =1 \,\forall i$}
In the following, we modify the work in \cite{haideralkim2019psmc} such that we guarantee the accommodation of errors in the presence of partially stuck cells such that $(c_i+e_i)_{ov}\geq s_i =1$.

It is important to note that the case where $s_i =1$ is of particular importance as this means, in multi-level PSMCs, all partially crystalline levels are reachable except the amorphous state \cite{wachterzeh2016codes}. 
\begin{const}
	\label{construction_1}
    Let $u\leq\text{min}\left\{n,\lfloor\frac{q-1}{2}\right\rfloor\}$. 
    	Suppose there is an $[n, k, d]_q$ code $\mycode{C}$ with a $k \times n$ generator matrix of the form
    \begin{align*}
    	\ve{G} = 
    	\begin{bmatrix}
    		\ve{G}_1\\
    		\ve{G}_0
    	\end{bmatrix}
    	=\begin{bmatrix} \ve{0}_{(n-r-1) \times 1} & \ve{I}_{n-r-1} & \ve{P}_{(n-r-1) \times r} \\ {1} & {\ve{1}_{n-r-1}} & {\ve{1}_{r}} \end{bmatrix},
    \end{align*}
    where $k=n-r$, $\ve{I}_{n-r-1}$ is the $(n-r-1) \times (n-r-1)$ identity matrix, $\ve{P} \in \mathbb{F}^{(n-r-1)\times (r)}_q$, and $\ve{1}_{\ell}$ is the all-one vector of length $\ell$.
 A PDMC can be obtained from the code $\mycode{C}$ using the encoder and the decoder as presented in \cite[Algorithm~1]{haideralkim2019psmc} and \cite[Algorithm~2]{haideralkim2019psmc} respectively, with slight modification in \cite[Algorithm~1]{haideralkim2019psmc} as follows:
   \begin{itemize}
       \item Step 2 finds
      \begin{align}    
       	           \left\{ v_b\in\mathbb{F}_q  \,\Big|\, \exists v_{b+1}=v_b+1 \, \Bigg[ 
          w_i\neq
       \begin{cases}
       	v_b, \\
       	v_{b+1}, 
       \end{cases}\forall\, i \in \ve{\phi}
           \Bigg]\right\}\nonumber
       \end{align} 
   $\mbox{ for all } b \in[\gamma] \mbox{ where }\gamma = q-u $.
       \item Step 3 
       takes $z_0\leftarrow -v_{b+1}$.
   \end{itemize}
   Then the encoder outputs $c_i\in\F_q\backslash \{0,q-1\}$.
\end{const}
\begin{thmmystyle}
	\label{Theorem_1}
	The coding model in Construction~\ref{construction_1} is an $(n, q^{n-r-1})$ ($u,t$)-PDMC.
\end{thmmystyle}  
\begin{proof}
	Recall that using \cite[Algorithm~1]{haideralkim2019psmc}, the encoder computes the  vector $(\ve{w} =\ve{m} \cdot \ve{G}_1) \in \F^n_q$ where the message vector $\ve{m} \in \mathbb{F}^{n-r-1}_q$, then adds the masking vector $\ve{d} = z_0 \cdot \ve{G}_0$ where $z_0 \in \mathbb{F}_q$ to output a codeword $\ve{c} \in \F^n_q$ that could be corrupted by $\ve{e}$ such that $\ve{y} = \ve{c} +\ve{e}$ (cf. Figure~\ref{fig:blockdiagramwithlegend}).
Since errors $\ve{e} \in\{0,1\}^n$ and zeros are not permitted in the partially defective positions,
the output codeword $\boldsymbol{c}$ at the $\ve{\phi}$ positions must fulfill
\begin{equation}
    \{1\leq c_i\leq q-2 \mid (c_i+e_j)_{ov}\geq s_i=1\},
    \label{masking}
\end{equation}
where $1\leq c_i + e_j\leq q-1$. 
Because $u\leq \lfloor\frac{q-1}{2}\rfloor$, there are $(\gamma = q-u>\lfloor\frac{q-1}{2}\rfloor)$ values of elements $v_b \in\mathbb{F}_q$ such that $\mbox{for all } i\in\ve{\phi} \mbox{ and for all }  b\in[\gamma]:w_i\neq v_b$.
Moreover, for $b\in[\gamma]$ there is at least a pair of consecutive $\{v_b \in \mathbb{F}_q\mid \exists v_{b+1}=v_b+1\}$. Thus, the encoder chooses $-v_{b+1}$ and obtains
\begin{gather*}
    c_i=w_i-v_{b+1} 
    \Rightarrow c_i = w_i-v_b-1.
\end{gather*}
Because $w_i\neq v_{b+1}$, $w_i-v_b-1 \neq 0 \Rightarrow w_i-v_b \neq 1$. 
Furthermore, since $w_i - v_b \neq 0$ as well, then $w_i-v_b-1 \in \F_q\backslash\{0,q-1\}$.
Hence, we obtain $1 \leq c_i \leq q-2$.  
So far we satisfied the masking condition.
Now for $e_j \in\{0,1\}$, $c_i  + (1 \text{ or } 0) \neq 0$.
Hence, $c_i+e_i \in\mathbb{F}_q\backslash\{0\}$ and \eqref{masking} is satisfied. 

As $\ve{c} \in \F^n_q$ in which $c_i\in\F_q\backslash \{0,q-1\}$ is a codeword in the code $\mycode{C}$ as well, the decoder decodes $\ve{y} = \ve{c} +\ve{e}$ using  \cite[Algorithm~2]{haideralkim2019psmc} to correct $t$ errors and retrieve $\ve{m}$.
\end{proof}

The masking capability of Theorem \ref{Theorem_1} is, however, half that of \cite[Theorem ~1]{haideralkim2019psmc} ($u\leq\lfloor\frac{q-1}{2}\rfloor$ as compared to $u \leq q-1$ respectively). 
Nevertheless, our method removes the impractical assumption ($c_i+e_j \neq 0$) in \cite{haideralkim2019psmc}. So we sacrifice some masking capability when relaxing this supposition. 

We introduce Corollary~\ref{Corollary_1} that reduces the gap such that up to $u \leq q-2$ partially defective-at-1 cells can be masked. 
\begin{cor}
\label{Corollary_1}
Construction~\ref{construction_1} 
can mask up to $u \leq q-2$ partially \pdmc\unskip-at-1 cells instead of $u\leq\lfloor\frac{q-1}{2}\rfloor$ if and only if at least a pair of consecutive elements $(v_b, v_{b+1})\in \F_q \mbox{ for } b \in [\gamma]$ (necessary condition) exist.
\end{cor}
\begin{proof}
Let there be $u=q-2$ 
partially \pdmc\unskip-at-1 cells. 
Then there must be at least \emph{two} choices of $v_b\in \F_q\backslash w_i$ (not necessarily consecutive) 
for $i \in \ve{\phi}$. 
Now, if by coincidence $\exists v_{b+1}=v_b+1$ due to the random augmented message vector, 
then the proof follows the proof of Theorem \ref{Theorem_1} while the encoder masks up to $q-2$ partially \pdmc\unskip-at-1 cells. 
\end{proof}
In the following, we prove the probability that such a pair of consecutive $(v_b, v_{b+1}) \in \F_q$ exists such that we can mask 
up to $u =n-1$
partially defective-at-1 cells while assuring non-zero appearances due to errors in the overlapped positions.
\begin{thmmystyle}\label{masking_with_prob}
    Assume $\ve{G}$ is as in Construction~\ref{construction_1} and let $q\leq \mid \ve{\phi} \mid \leq n-r-1,$
    $s_i =1$ for all $i \in \ve{\phi}$. Let 
    the columns of $\ve{G}$ labeled by $\ve{\phi}$ be linearly independent.  The masking probability of the codeword such that $c_i \in \F_q\backslash\{0,q-1\}$ is
    \begin{equation}\label{prob_masking}
        \P(\{v_b,v_{b+1}\},q,u) = \frac{q}{{q \choose 2}} \times \left(1- \dfrac{ \sum_{i=0}^{q-2} (-1)^i {q \choose i} (q-i)^{u}}{q^{u}} \right) 
    \end{equation}
    for a message $\ve{m} \in \F^{n-r-1}_q$ that is drawn uniformly at random. 
\end{thmmystyle}
 \begin{proof}
 	We know $\mid \ve{\phi} \mid =u$ and we rewrite the probability in \cite[Theorem 3]{haideralkim2019psmc} as such:
  \begin{equation}\label{prob_one_values_not_occur}
   \P(q,u) = 1- \dfrac{ \sum_{i=0}^{q-1} (-1)^i {q \choose i} (q-i)^{u}}{q^{u}}.
   \end{equation}
    Then we follow the proof of \cite[Theorem 1]{haideralkim2019psmc} with slight modification as follows. 
   As at least \emph{two} values out of $\F_q$ elements must be excluded to satisfy \eqref{masking}, the upper sum limit in \eqref{prob_one_values_not_occur} reduces to $q-2$ and the probability becomes
      \begin{equation}\label{prob_two_values_not_occur}
        \P(q,u) = 1- \dfrac{ \sum_{i=0}^{q-2} (-1)^i {q \choose i} (q-i)^{u}}{q^{u}}. 
    \end{equation}
   In \eqref{prob_two_values_not_occur}, the inclusion-exclusion
technique~\cite{Stanley1986} determines the relative number of vectors in $\F^u_q$ that do not include at least two field elements. 
   
 As these two field elements in $\F_q$ are not necessarily consecutive, we need to find the probability of having successive pairs. To find the number of unique pairs where the pairs are subject to the commutative property, we calculate $ {q \choose 2}$. Among $ {q \choose 2}$ pair sets, there are exactly $q$ sets that have two consecutive entries $\{v_b,v_{b+1}\}$ (including the pair $\{q-1, 0\}$ which has elements that are sequential due to modulo operation). Hence, the probability of consecutive $\{v_b,v_{b+1}\} \subset \F_q$ occurring is
     \begin{equation}\label{consecutive_sets}
        \P(\{v_b,v_{b+1}\},q) = \frac{q}{{q \choose 2}}. 
    \end{equation}
    Combining \eqref{prob_two_values_not_occur} and \eqref{consecutive_sets} gives \eqref{prob_masking}.
 \end{proof}  
\begin{examplex}
Let $q =3$, $n=8$, $r=0$. The probability of masking $u=n-1$ partially defective-at-1 cells such that the output vector satisfies \eqref{masking}
due to the existence of two consecutive elements $\{v_b,v_{b+1}\} \subset \F_q$ is
$\P(\{v_b,v_{b+1}\},3,7) = 0.175$.
\end{examplex}
We would like to mention here that
instead of the construction using a generator matrix by Theorem~\ref{Theorem_1},
we could also use an alternative method using \emph{the partitioned cyclic code construction} similar to \cite[Theorem~2]{haideralkim2019psmc}. 
Hence, an all-one polynomial corresponds to the all-one vector in Construction~\ref{construction_1}.
Then, we can mask $u\leq\lfloor\frac{q-1}{2}\rfloor$ partially \pdmc\unskip-at-1 cells and fulfill \eqref{masking}
by finding $v_{b+1}$ (as stated in the proof of Theorem~\ref{Theorem_1}), and multiplying $v_{b+1}$ by the all-one polynomial in Step~7 in \cite[Algorithm 3]{haideralkim2019psmc} where the coefficients of the output polynomial are in $\F_q\backslash \{0,q-1\}$ (cf. Step~8 in \cite[Algorithm 3]{haideralkim2019psmc}). 
\begin{rem}\label{Rem_extended}
The assumption that $\ve{e}$ is of coordinates $\in\{0,1\}^n$ could be extended such that $\ve{e} \in \{0,x\}^n$, where $x \in \F_q\backslash{0}$. 
Then consecutive elements $\{v_b,v_{b+1}\} \subset \F_q$ such that $v_{b+1}$ satisfies \eqref{masking} also exist.
\end{rem}
\section{Codes for Preventing Zeros at Partially \Pdmc Positions Due to Random errors , $\lfloor\frac{q-1}{2}\rfloor+1\leq u\leq n, s_i =1 \,\forall i$}
In this part of the paper, we first provide Lemma~\ref{Lem_2} as our bedrock for the generalization to mask any number of $u$ cells while preventing zeros in the partially stuck positions due to random errors.
\begin{lem}
\label{Lem_2}(Masking Only)
Let $q$ be a prime power.
Assume the matrix $\ve{H} = (H_{i,j})^{i \in [2\kappa]}_{j\in[n]}$ 
is given.
Let each block length of any $2\kappa \times u$ submatrix denoted by $\boldsymbol{H}^{(u)}$ be $\leq \lfloor\frac{q-1}{2}\rfloor$
such that:	
\\
$ \text{RRE}(\boldsymbol{H}^{(u)}) $ 
\scalebox{0.75}{
$
	\setlength{\arraycolsep}{2pt}
\renewcommand\arraystretch{0.2}{
	=\left[
\begin{array}{ccccc}
	\underbrace{\begin{array}{rrrr}
			1 &\fullcirc[0.5ex]  & \dots & \fullcirc[0.5ex] 
	\end{array} }_{\leq \lfloor\frac{q-1}{2}\rfloor}   & \begin{array}{rrrr}
		\emptycirc[0.5ex] & \emptycirc[0.5ex] & \dots & \emptycirc[0.5ex] 
	\end{array}       & \dots   & \dots   & \begin{array}{rr}
		\dots & \emptycirc[0.5ex] 
	\end{array} \\
	& \underbrace{\begin{array}{rrrr}
			1 &\fullcirc[0.5ex]  & \dots & \fullcirc[0.5ex] 
	\end{array} }_{\leq \lfloor\frac{q-1}{2}\rfloor}    & \begin{array}{rrrr}
		\emptycirc[0.5ex] & \emptycirc[0.5ex] & \dots & \emptycirc[0.5ex] 
	\end{array} & \dots & \begin{array}{rr}
		\dots & \emptycirc[0.5ex] 
	\end{array} \\ 
	& & \underbrace{\begin{array}{rrrr}
			1 &\fullcirc[0.5ex]  & \dots & \fullcirc[0.5ex] 
	\end{array} }_{\leq \lfloor\frac{q-1}{2}\rfloor}    & \begin{array}{rr}
		\emptycirc[0.5ex] & \dots 
	\end{array}    & \begin{array}{rr}
		\dots & \emptycirc[0.5ex] 
	\end{array} \\ 
	&  &  & \ddots & \vdots \\ 
	\bigzero &  &  &   & \underbrace{\begin{array}{rrrr}
			1 &\fullcirc[0.5ex]  & \dots & \fullcirc[0.5ex] 
	\end{array} }_{\leq \lfloor\frac{q-1}{2}\rfloor} \\ 
\end{array} 
	\right],
}$
}
\vspace{0.2cm}
\\
where $\fullcirc[0.5ex]\in\mathbb{F}_q^*$ and $\emptycirc[0.5ex]\in\mathbb{F}_q$. 
 Then, there exists a $u$-\psmc over $\mathbb{F}_q$ of length $n$ and redundancy $r = 2\kappa$ that has the coefficients at $u$ positions in $\mathbb{F}_q\backslash\{0,q-1\}$.
\end{lem}
\begin{proof}
The proof closely resembles that of \cite[Theorem ~7]{wachterzeh2016codes}. To that end, we assume w.l.o.g. that each block length is precisely $\lfloor\frac{q-1}{2}\rfloor$; it is clear that the same principle applies if it is shorter. 
Pertaining to the $i$-th block for $i \in [2\kappa]$, the encoder picks $z_i = v_{b+1} \in\mathbb{F}_q$ in its corresponding $i$-th block to come up with the vector $\ve{z} = (z_0,z_1,\dots,z_{2\kappa-1})$ such that
\begin{align}
   \label{ith_block}
\nonumber &z_i \cdot H_{i,j}+(w_j+z_{i-1} \cdot H_{i-1,j}+\dots+z_0 \cdot H_{0,j})\in\mathcal{F}_i,\\ \nonumber&\text{where } \mathcal{F}_i\subset\mathbb{F}_q\text{ with }\vert\mathcal{F}_i\vert\leq \Big\lfloor\frac{q-1}{2}\Big\rfloor,\forall i\in[2\kappa] \mbox{ and}\\ 
& j\in\left[i\Big\lfloor\frac{q-1}{2}\Big\rfloor,(i+1)\Big\lfloor\frac{q-1}{2}\Big\rfloor\right].
\end{align}
Because there are (at most) $\lfloor\frac{q-1}{2}\rfloor$ constraints in (\ref{ith_block}) and $q$ possible values for $z_i$, there are at least two possible successive values $v_{b}$ and $v_{b+1}$  in each of the $i\in[2\kappa]$ blocks. 
Since $\ve{z}$ is not unique, the encoder first finds $z_0$ that masks the first $0$-th block of size at most $\lfloor\frac{q-1}{2}\rfloor$, then $z_1$ to mask the next block, until selecting $z_{2\kappa-1}$ that deals with the last block while satisfying \eqref{ith_block}. Hence, 
\[(z_0,z_1,\dots,z_{2\kappa-1})\cdot \boldsymbol{H}^{(u)}\]
masks the $i\lfloor\frac{(q-1)}{2}\rfloor$-th to $((i+1)\lfloor\frac{(q-1)}{2}\rfloor-1)$-th partially stuck-at-$1$ cells while satisfying \eqref{masking} in each $i$-th block such that the coefficients at $u$ positions in $\mathbb{F}_q\backslash\{0,q-1\}$.
\newline\indent This principle encompasses all $2\kappa$ blocks of $\lfloor\frac{q-1}{2}\rfloor$ cells and column permutations of $\text{RRE}(\boldsymbol{H}^{(u)})$ are clearly not an issue to this approach.
\end{proof}
Now, by replacing the all-one vector in Construction~\ref{construction_1} with a parity-check matrix $\boldsymbol{H}$ from Lemma~\ref{Lem_2}, we introduce Theorem~\ref{Theorem_3} to be able to mask $u\leq n$ partially \pdmc\unskip-at-1 cells while satisfying \eqref{masking}.
\begin{const}
	\label{construction_2}
	 Let $u \leq n$ and $s_i=1$. Suppose there is an $[n, k, d]_q$ code $\mycode{C}$ with a $k \times n$ generator matrix $\ve{G}$ from Construction~\ref{construction_1} such that $\boldsymbol{H} \in \mathbb{F}^{2\kappa \times n}_q$ from Lemma~\ref{Lem_2} replaces $\ve{G_0}$.
   A PDMC can be obtained from the code $\mycode{C}$ using the encoder and the decoder as presented in \cite[Algorithms~5 and~6]{haideralkim2019psmc}.
   Then the encoder outputs $c_i\in\F_q\backslash \{0,q-1\}$.
\end{const}
\begin{thmmystyle}
    \label{Theorem_3}
  	The coding model in Construction~\ref{construction_2} is an $(n, q^{n-r-2\kappa})$ ($u,t$)-PDMC.
\end{thmmystyle}
\begin{proof}
\textbf{Masking:} In order to mask each constituent "block" of $\text{RRE}(\boldsymbol{H}^{(u)})$ following Lemma~\ref{Lem_2} such that $c_i\in\mathbb{F}_q\backslash\{0,q-1\}$,
the encoder picks $z_i \in \mathcal{F}_i \subset \mathbb{F}_q$ to be unequal to $w_i \in \mathbb{F}_q \backslash \mathcal{F}_i $.
Then, $\ve{z} \cdot \H^{(u)}$ masks all $u \leq n$ positions and fulfills \eqref{masking}.

\textbf{Error Correction:}  
We use the full rank matrix $\boldsymbol{G}$ of a code that has enough minimum distance to correct at most $t$ errors.  
\end{proof}
\begin{rem}
By halving the maximum length of each block, the number of rows of RRE($\boldsymbol{H}^{(u)}$) in Lemma~\ref{Lem_2} is at least twice that of \cite[Theorem ~7]{wachterzeh2016codes} for fixed $q$ and $u$. This implies we require double the redundancy for masking while correcting the error in the \emph{overlapping} model compared to the masking only case given in \cite[Theorem ~7]{wachterzeh2016codes}.
\end{rem}
We improve Theorem~\ref{Theorem_3} by introducing Theorem~\ref{Theorem_4} such that we reduce the required redundancy for masking in a similar way given in \cite[Theorem ~8]{wachterzeh2016codes} through Lemma~\ref{lemma_2}. 
\begin{lem}
    \label{lemma_2} (Masking Only)
    Let $u\leq\text{min}\{\lfloor\frac{q-1}{2}\rfloor+d-2,n\}$.
    Then, the existence of a $u$-\psmc over $\mathbb{F}_q$ of length $n$ and redundancy $l \leq 2\kappa$ is guaranteed.
\end{lem}
\begin{proof}
According to Lemma~\ref{Lem_2}, $\text{RRE}(\boldsymbol{H}^{(u)})$ has block size $\leq \lfloor\frac{q-1}{2}\rfloor$. Hence, following the proof of \cite[Theorem ~8]{wachterzeh2016codes}, we obtain $u-d+2\leq\lfloor\frac{q-1}{2}\rfloor$.
Therefore, we can mask any $u$ partially \pdmc\unskip-at-1 cells with $\boldsymbol{H} \in \mathbb{F}^{l \times n}_q$ as
a special case of $\boldsymbol{H} \in \mathbb{F}^{2\kappa \times n}_q$ from Lemma~\ref{Lem_2}, where $l \leq 2\kappa$.
\end{proof}
\begin{const}\label{construction_3}
 Let $u\leq\lfloor\frac{q-1}{2}\rfloor+d_0-2$ and $s_i=1$. 
Suppose there is an $[n, k, d]_q$ code $\mycode{C}$ with a $k \times n$ generator matrix $\ve{G}$ of the form stated in Construction~\ref{construction_2}. Let a subcode $\mathcal{C}_0^\perp$ of $\mathcal{C}$ be generated by a systematic parity-check matrix $\boldsymbol{H}_0 \in \mathbb{F}^{l \times n}_q$ that replaces $\boldsymbol{H} \in \mathbb{F}^{2\kappa \times n}_q$, whose dual code is an $[n,n-l,d_0 \geq u-\lfloor\frac{q-1}{2}\rfloor+2]_{q}$ code $\mathcal{C}_0$.
    
A PDMC can be obtained from the code $\mycode{C}$ using the encoder and the decoder as stated in \cite[Algorithm~5]{haideralkim2019psmc} and \cite[Algorithm~6]{haideralkim2019psmc}  respectively, with slight modification in \cite[Algorithm~5]{haideralkim2019psmc} as follows:
   \begin{itemize}
       \item Step 2 finds $\ve{z}$ as shown in the proof of Lemma~\ref{Lem_2}.
   \end{itemize}
   Then the encoder outputs $c_i\in\F_q\backslash \{0,q-1\}$.
\end{const}
\begin{thmmystyle}
    \label{Theorem_4}
   	The coding model in Construction~\ref{construction_3} is an $(n, q^{n-r-l})$ ($u,t$)-PDMC.
\end{thmmystyle}
\begin{proof}
For masking, we use $\boldsymbol{H}_0 \in \F_q^{l\times n}$ in Lemma~\ref{lemma_2}.
As pointed out in the proof of Lemma~\ref{Lem_2}, the encoder chooses $z_i = v_{b+1}$ for each of the $i\in[l]$ blocks within $\boldsymbol{H}_0$ such that (\ref{ith_block}) is fulfilled. Therefore, all $u$ partially \pdmc\unskip-at-1 cells can successfully be masked and the output $c_i\in\F_q\backslash \{0,q-1\}$ for $i \in \ve{\phi}$.
The error correction part of the proof follows Theorem~\ref{Theorem_3}.
\end{proof}
\section{Correcting $(u+t)$ Errors}
In the following section, we provide a simple proposition 
to prove that masking $u$ cells directly by the error correction capability of the code (very uncomplicated) works for a few code parameters as this was not clearly mentioned in \cite{haideralkim2019psmc}.
\begin{prop}
    \label{Proposition_3}
    Let $u,t, n, \,\mbox{and}\,q$ be positive integers.
    Let $
    \{0\leq u,t \leq n\mid u+t < n\}$. Let the encoder introduce artificial errors in $\ve{\phi}$ such that $(c_i+e_j)_{ov} \geq s_i$ for $i \in \ve{\phi}$ and $j \in \Psi$.
    Assume there is an $[n,k,d\geq 2(t+u)+1]_q$ code $\mathcal{C}$. 
    Then there is an $(n, q^k)$ $(u,t)$-PDMC. 
\end{prop}
\begin{proof}
Since the encoder deliberately introduces errors in $\ve{\phi}$, i.e., setting the values in each $i \in \ve{\phi}$ to $1$, all $\ve{\phi}$ locations are masked (non-zero). Now, as code $\mathcal{C}$ has enough minimum distance $d\geq 2(t+u)+1$, it can correct the channel errors $t$ and $u$ artificial errors.  
\end{proof}
Section~\ref{Comparisons} compares Proposition~\ref{Proposition_3} with the complicated 
error correcting and masking code construction suggested in \cite[Theorem~1]{haideralkim2019psmc} for a few code parameters.
\section{Comparison}\label{Comparisons}
Consider a $[114,8,79]_7$ BCH code of rate $0.0701$. 
This code can correct up to a total of $u+t=39$ errors by Proposition~\ref{Proposition_3}.
On the contrary, application of \cite[Theorem~1]{haideralkim2019psmc} using a $[114,9,67]_7$ BCH code of rate $0.07889$ for $u=6$ achieves the same previous rate as follows. 
Since we require a single symbol for masking $u=6$ by \cite[Theorem~1]{haideralkim2019psmc}, that is
$\frac{1}{114}=0.00879$ for this code, we obtain the rate $0.07889 - 0.00877 = 0.0701$, which is exactly the rate of the former BCH code, and we are also able to correct $u+t = 6+33 =39$ errors. 
We conclude that directly applying Proposition~\ref{Proposition_3} could be 
as effective as employing \cite[Theorem~1]{haideralkim2019psmc}.  
However, application of Proposition~\ref{Proposition_3} for most code parameters is not practical as the required check symbols are considerably large to obtain a code $\mathcal{C}$ with enough minimum distance $d\geq 2(t+u)+1$ to correct both $t+u$ errors. 
For $u \geq q$ and by the prior argument, we infer also that Proposition~\ref{Proposition_3} is much worse than \cite[Theorem~4]{haideralkim2019psmc} and Construction~\ref{construction_2}. 
\section{Conclusion}
In this work, we have first derived the probability of intersections between $t$ errors and $u$ cells.
Then we proved that the probability of such overlaps resulting in zeros in the positions of the partial defects is non-trivial (i.e., for $q=3$, one-third of the overlapping cases leads to zeros in the partially defective coordinates). Hence, we have modified the constructions given in \cite{haideralkim2019psmc} such that the encoded vector never attains $0$ or $q-1$ at the partially defective positions. 
Such a modification ensures that if the support of the error vector $\{\ve{e}\in \{0,1\}^n \mid wt(\ve{e}) \leq t\}$ occurs at the partially defective locations, the resulting vector's coefficients at these positions are non-zero. As two values from the set $[q]$ in the output vector are forbidden, this modification either costs more redundancy for masking while considering overlapping errors or handles less masked cells.
We also showed that it is possible, for a few code parameters, to directly coincide with the partially defective cells by the error-correcting code capability instead of sophisticated error-correcting and defect-masking schemes. 
\bibliographystyle{IEEEtran}

\begin{thebibliography}{9}
	\bibitem{haideralkim2019psmc}
	H. Al Kim, S. Puchinger, and A. Wachter-Zeh, “Error Correction for Partially Stuck Memory Cells,”\emph{in 2019 XVI International Symposium "Problems of Redundancy in Information and Control Systems" (REDUNDANCY)}, Moscow, Russia, 2019, pp. 87–92.
	\bibitem{Gleixner2009}
    B. Gleixner, F. Pellizzer, and R. Bez, “Reliability characterization of phase change memory,” \emph{in 2009 10th Annual Non-Volatile Memory Technology Symposium (NVMTS)}. IEEE, 2009, pp. 7–11.
   \bibitem{Kim2005}  
    K. Kim and S. J. Ahn, “Reliability investigations for manufacturable high density pram,” in 2005 IEEE International Reliability Physics Symposium, 2005. Proceedings. 43rd Annual. IEEE, 2005, pp. 157–162.
    \bibitem{Lee2009}
     S. Lee, J.-h. Jeong, T. S. Lee, W. M. Kim, and B.-k. Cheong, “\emph{A study on the failure mechanism of a phase-change memory in write/erase cycling},”IEEE Electron Device Letters, vol. 30, no. 5, pp. 448–450, 2009.
     \bibitem{Pirovano2004}
      A. Pirovano, A. Redaelli, F. Pellizzer, F. Ottogalli, M. Tosi, D. Ielmini, A. L. Lacaita, and R. Bez, “\emph{Reliability study of phase-change nonvolatile memories},” IEEE Transactions on Device and Materials Reliability, vol. 4, no. 3, pp. 422–427, 2004.
	\bibitem{heegard1983partitioned}
	C. Heegard, “Partitioned Linear Block Codes for Computer Memory with ’Stuck-at’ Defects,” \emph{IEEE Trans. Inf. Theory}, vol. 29, no. 6, pp. 831–842, 1983.
	\bibitem{wachterzeh2016codes}
	A. Wachter-Zeh and E. Yaakobi, "Codes for Partially Stuck-at Memory Cells," \emph{IEEE Trans. Inf. Theory}, vol 62, no. 2, pp. 639-654, 2016.
	\bibitem{ASolomon}
	A. Solomon and Y. Cassuto, "Error-Correcting WOM Codes: Concatenation and Joint Design," in IEEE Transactions on Information Theory, vol. 65, no. 9, pp. 5529-5546, Sept. 2019, doi: 10.1109/TIT.2019.2917519.
    \bibitem{Stanley1986}
	R. P. Stanley, Sieve Methods. Boston, MA: Springer US, 1986, pp. 64–95.
	[Online]. Available: https://doi.org/10.1007/978-1-4615-9763-6\_2
\end{thebibliography}

\end{document}